\documentclass[11pt]{article}
\usepackage{graphicx}
\usepackage{tabularx}
\usepackage{url}
\usepackage{amsthm}
\usepackage{amsmath}
\usepackage{amsfonts}
\usepackage[labelfont=bf]{caption}

\theoremstyle{plain}
\newtheorem{theorem}{Theorem}
\newtheorem{lemma}{Lemma}

\theoremstyle{definition}
\newtheorem{definition}{Definition}

\newtheorem*{ex}{Example}
\theoremstyle{remark}

\date{}

\begin{document}

\title{Comparison of Deterministic and Nondeterministic Decision Trees
for Decision Tables with Many-valued Decisions from Closed Classes}

\author{Azimkhon Ostonov and Mikhail Moshkov \\
Computer, Electrical and Mathematical Sciences \& Engineering Division \\ and Computational Bioscience Research Center\\
King Abdullah University of Science and Technology (KAUST) \\
Thuwal 23955-6900, Saudi Arabia\\ \{azimkhon.ostonov,mikhail.moshkov\}@kaust.edu.sa
}

\maketitle

\begin{abstract}
In this paper, we consider classes of decision tables with many-valued decisions closed relative to removal of attributes (columns) and changing sets of decisions assigned to rows. For tables from an arbitrary closed class, we study a function $\mathcal{H}^{\infty}_{\psi ,A}(n)$ that characterizes the dependence in the worst case of the minimum complexity of deterministic decision trees on the minimum complexity of nondeterministic decision trees. Note that nondeterministic decision trees for a decision table can be interpreted as a way to represent an arbitrary system of true decision rules for this table that cover all rows.
We indicate the condition for the function $\mathcal{H}^{\infty}_{\psi ,A}(n)$ to be defined everywhere. If this function is everywhere defined, then it is either bounded from above by a constant or is greater than or equal to $n$ for infinitely many $n$. In particular, for any nondecreasing function  $\varphi$ such that $\varphi (n)\geq n$ and $\varphi (0)=0$, the function $\mathcal{H}^{\infty}_{\psi ,A}(n)$ can grow between $\varphi (n)$ and $\varphi (n)+n$. We indicate also conditions for the function $\mathcal{H}^{\infty}_{\psi ,A}(n)$ to be bounded from above by a polynomial on $n$.
\end{abstract}


\section{Introduction\label{7S0}}

In this paper, we consider closed classes of decision tables with many-valued decisions.  For tables from an arbitrary closed class,
we study a function that characterizes the dependence in the worst case of the minimum complexity of deterministic  decision trees on  the minimum complexity of nondeterministic decision trees.

A decision table with many-valued decisions is a rectangular table in which
columns are labeled with attributes, rows are pairwise different and each
row is labeled with a nonempty finite set of decisions. Rows are interpreted
as tuples of values of the attributes. For a given row, it is required to
find a decision from the set of decisions attached to the row. To this end,
we can use the following queries: we can choose an attribute and ask what is
the value of this attribute in the considered row. We study two types of
algorithms based on these queries: deterministic and nondeterministic
decision trees. One can interpret nondeterministic decision trees for a
decision table as a way to represent an arbitrary system of true decision
rules for this table that cover all rows. We consider so-called bounded
complexity measures that characterize the time complexity of decision
trees, for example, the depth of decision trees.

Decision tables with many-valued decisions often appear in data analysis,
where they are known as multi-label decision tables \cite{Boutell04,Vens08,Zhou12}.
Moreover, decision tables with many-valued decisions are common in such areas as
combinatorial optimization, computational geometry, and fault diagnosis, where they
are used to represent and explore problems \cite{AlsolamiACM20,MoshkovZ11}.

Decision trees \cite%
{AbouEishaACHM19,AlsolamiACM20,BreimanFOS84,Moshkov05,Moshkov20,Quinlan93,RokachM07}
and decision rule systems \cite%
{BorosHIK97,BorosHIKMM00,ChikalovLLMNSZ13,FurnkranzGL12,MPZ08,MoshkovZ11,Pawlak91,PawlakS07}
are widely used as a means for knowledge representation, as classifiers that
predict decisions for new objects, and as algorithms for solving various
problems of fault diagnosis, combinatorial optimization, etc. Decision trees
and rules are among the most interpretable models for classifying and
representing knowledge \cite{Molnar22}. The study of relationships between these
two models is an important task of computer science. This paper is precisely devoted to the study of these relationships.

We explore classes of decision tables with many-valued decisions closed under
removal of columns (attributes) and changing the decisions (really, sets of
decisions). The most natural examples of such
classes are closed classes of decision tables generated by information
systems  \cite{Pawlak81}. An information system consists of a set of objects
(universe) and a set of attributes (functions) defined on the universe and
with values from a finite set. A problem with many-valued decisions over an
information system is specified by a finite number of attributes that divide the universe into
nonempty domains in which these attributes have fixed values. A nonempty
finite set of decisions is attached to each domain. For a given object from
the universe, it is required to find a decision from the set attached to the
domain containing this object.

A decision table with many-valued decisions
corresponds to this problem in a natural way: columns of this table are
labeled with the considered attributes, rows correspond to domains and are
labeled with sets of decisions attached to domains. The set of decision
tables corresponding to problems with many-valued decisions over an information system forms a closed class
generated by this system. Note that the family of all closed classes is
essentially wider than the family of closed classes generated by information
systems. In particular, the union of two closed classes generated by two
information systems is a closed class. However, generally, there is no an
information system that generates this class.

Let $A$ be a class of decision tables with many-valued decisions closed under removal of
columns and changing of decisions, and $\psi$ be a bounded complexity
measure. In this paper, we study the function $\mathcal{H}^{\infty}_{\psi ,A}(n)$, which  characterizes for decision tables from $A$ the growth in the worst case of the minimum complexity of deterministic decision trees with the growth of the minimum complexity of nondeterministic decision trees.

We indicate the condition for the function $\mathcal{H}^{\infty}_{\psi ,A}(n)$ to be defined
everywhere. If this condition is satisfied, then the considered function is either bounded from above by a constant, or is greater than or equal to $n$ for infinitely many $n$. In particular, for any nondecreasing function  $\varphi$ such that $\varphi (n)\geq n$ and $\varphi (0)=0$, the function $\mathcal{H}^{\infty}_{\psi ,A}(n)$ can grow between $\varphi (n)$ and $\varphi (n)+n$. We indicate also conditions for the function $\mathcal{H}^{\infty}_{\psi ,A}(n)$ to be bounded from above by a polynomial on $n$.

The results obtained allow us to significantly better understand how the complexity of deterministic and nondeterministic decision trees for decision tables with many-valued decisions relate. Note that, by comparing the complexity of deterministic and nondeterministic decision trees, we are essentially comparing the time complexity of sequential and parallel (one processor per rule) classification algorithms for decision tables with many-valued decisions.

We also established the following significant difference between conventional decision tables \cite{Ostonov23d} in which one decision is attached to each row and decision tables with many-valued decisions.
A decision table with $n$ columns filled with zeros and ones is called complete if it has $2^n$ pairwise different rows. One can show that, for a conventional complete decision table, the minimum depth of a deterministic decision tree  does not exceed the square of the minimum depth of a nondeterministic decision tree. This is a simple generalization of the results obtained for Boolean functions in \cite{BlumI87,HartmanisH87,Tardos89} (see also \cite{BuhrmanW02}). In the present paper, we construct an infinite sequence of complete decision tables with many-valued decisions for which the minimum depth of deterministic decision trees does not bounded from above by a constant but the minimum depth of nondeterministic decision trees is at most $3$ -- see Lemmas \ref{7M4} and \ref{7M6}.

Note that there is a similarity between some results obtained in this paper for closed classes of decision tables and  results from the book \cite{Moshkov20} obtained for problems over information systems. However, the results of the present paper are more general since the family of all closed classes is essentially wider than the family of closed classes generated by information systems.

This paper is an extended version of a part of the conference  paper \cite{Ostonov23c}, which does not contain proofs.

The present paper consists of six sections. In Section \ref{7S1}, main definitions and
notation are considered. In Section \ref{7S2}, we provide the main results.
In Sections \ref{7S4}-\ref{7S5}, we prove the main results. Section \ref{7S6} contains short conclusions.

\section{Main Definitions and Notation\label{7S1}}

Denote $\omega =\{0,1,2,\ldots \}$, $\mathcal{P}(\omega )$ the set of
nonempty finite subsets of the set $\omega $ and, for any $k\in \omega
\setminus \{0,1\}$, denote $E_{k}=\{0,1,\ldots ,k-1\}$. Let $P=\{f_{i}:i\in
\omega \}$ be the set of \emph{attributes} (really, names of attributes). Two
attributes $f_{i},f_{j}\in P$ are considered \emph{different} if $i\neq j$.

\subsection{Decision Tables\label{7S1.1}}

First, we define the notion of a decision table.

\begin{definition}
Let $k\in \omega \setminus \{0,1\}$. Denote by $\mathcal{M}_{k}^{\infty }$
the set of rectangular tables filled with numbers from $E_{k}$ in each of
which rows are pairwise different, each row is labeled with a set from $%
\mathcal{P}(\omega )$ (set of decisions), and columns are labeled with
pairwise different attributes from $P$. Rows are interpreted as tuples of
values of these attributes. Empty tables without rows belong also to the set
$\mathcal{M}_{k}^{\infty }$. We will use the same notation $\Lambda $ for
these tables. Tables from $\mathcal{M}_{k}^{\infty }$ will be called
\emph{decision tables with many-valued decisions} (\emph{decision tables}).
\end{definition}
\index{Decision table with many-valued decisions}

For a table $T\in \mathcal{M}_{k}^{\infty }$, we denote by $%
\Delta (T)$ the set of rows of the table $T$ and by $\Pi (T)$ we denote the
intersection of sets of decisions attached to rows of $T$. Decisions from $%
\Pi (T)$ are called \emph{common decisions} for the table $T$.
\index{Decision table with many-valued decisions!common decision}

\begin{ex}
Figure \ref{7fig1} shows a decision table from $\mathcal{M}_{2}^{\infty }$.
\end{ex}

\begin{figure}[tbp]
\begin{minipage}[c]{1.0\textwidth}
\begin{center}
\begin{tabular}{ |ccc|c| }
 \hline
 $f_{2}$ & $f_{4}$ & $f_{3}$ &\\
  \hline
 1 & 1 & 1 & $\{1\}$ \\
 0 & 1 & 1 & $\{0, 1, 2\}$ \\
 1 & 1 & 0 & $\{1, 3\}$ \\
 0 & 0 & 1 & $\{2\}$ \\
 1 & 0 & 0 & $\{3\}$ \\
 0 & 0 & 0 & $\{2, 3\}$ \\
 \hline
\end{tabular}
\end{center}
\end{minipage}
\caption{Decision table from $\mathcal{M}_{2}^{\infty}$}
\label{7fig1}
\end{figure}

Denote by $\mathcal{M}_{k}^{\infty c}$ the set of tables from $%
\mathcal{M}_{k}^{\infty }$ in each of which there exists a common decision.
Let $\Lambda \in $ $\mathcal{M}_{k}^{\infty c}$.

Let $T$ be a nonempty table from $\mathcal{M}_{k}^{\infty }$. Denote by $\operatorname{At}(T)
$ the set of attributes attached to columns of the table $T$. We denote by $\Omega _{k}(T)$ the set of finite words over the alphabet
$\{(f_i,\delta):f_i \in \operatorname{At}(T), \delta \in E_k\}$ including the empty word $\lambda $.
For any $\alpha \in \Omega _{k}(T)$, we now define a subtable $T\alpha$ of the table $T$. If $\alpha = \lambda$, then $T\alpha =T$. If $\alpha \neq \lambda$ and $\alpha =(f_{i_{1}},\delta
_{1})\cdots (f_{i_{m}},\delta _{m})$, then $T\alpha $ is the table obtained from $T$
by removal of all rows that do not satisfy the following condition: in
columns labeled with attributes $f_{i_{1}},\ldots ,f_{i_{m}}$, the row has
numbers $\delta _{1},\ldots ,\delta _{m}$, respectively.

We now define two operations on decision tables: removal of columns and
changing of decisions. Let $T\in \mathcal{M}_{k}^{\infty }$.

\begin{definition}
\emph{Removal of columns.} Let $D\subseteq \operatorname{At}(T)$. We remove from $T$ all
columns labeled with the attributes from the set $D$. In each group of rows
equal on the remaining columns, we keep the first one. Denote the obtained
table by $I(D,T)$. In particular, $I(\emptyset ,T)=T$ and $I(\operatorname{At}(T),T)=\Lambda
$. It is obvious that $I(D,T)\in $ $\mathcal{M}_{k}^{\infty }$.
\end{definition}

\begin{definition}
	\emph{Changing of decisions.} Let $\nu :E_{k}^{\left\vert \operatorname{At}(T)\right\vert
}\rightarrow \mathcal{P}(\omega )$ (by definition, $E_{k}^{0}=\emptyset $).
For each row $\bar{\delta}$ of the table $T$, we replace the set of
decisions attached to this row with $\nu (\bar{\delta})$. We denote the
obtained table by $J(\nu ,T)$. It is obvious that $J(\nu ,T)\in $ $\mathcal{M%
}_{k}^{\infty }$.
\end{definition}

\begin{definition}
Denote $\left[ T\right] =\{J(\nu ,I(D,T)):D\subseteq \operatorname{At}(T),\nu
:E_{k}^{\left\vert \operatorname{At}(T)\setminus D\right\vert }\rightarrow \mathcal{P}%
(\omega )\}$. The set $\left[ T\right] $ is the \emph{closure of the table} $%
T$ under the operations of removal of columns and changing of decisions.
\end{definition}

\begin{ex}
Figure \ref{7fig2} shows the table $J(\nu ,I(D,T_0))$, where $T_0$ is the table
shown in Figure   \ref{7fig1}, $D=\{f_{4}\}$ and $\nu (x_{1},x_{2})=\{\min(x_{1},x_{2}), \max(x_{1},x_{2})\}$.
\end{ex}

\begin{figure}[tbp]
\begin{minipage}[c]{1.0\textwidth}
\begin{center}
\begin{tabular}{ |cc|c| }
 \hline
 $f_{2}$ & $f_{3}$ &\\
  \hline
 1 & 1 & $\{1\}$ \\
 0 & 1 & $\{0, 1\}$ \\
 1 & 0 & $\{0, 1\}$ \\
 0 & 0 & $\{0\}$ \\
 \hline
\end{tabular}
\end{center}
\end{minipage}
\caption{Decision table obtained from the decision table shown in Figure
\protect\ref{7fig1} by removal of a column and changing of decisions}
\label{7fig2}
\end{figure}

\begin{definition}
Let $A\subseteq \mathcal{M}_{k}^{\infty }$ and $A\neq \emptyset $. Denote $%
\left[ A\right] =\bigcup_{T\in A}\left[ T\right] $. The set $\left[ A\right]
$ is the \emph{closure} of the set $A$ under the considered two operations.
The class (the set) of decision tables $A$ will be called a \emph{closed
class} if $\left[ A\right] =A$.
\end{definition}
\index{Decision table with many-valued decisions!closed class}

A closed class of decision tables will be called \emph{nontrivial} if it contains nonempty decision tables.
\index{Decision table with many-valued decisions!closed class!nontrivial}

Let $A_{1}$ and $A_{2}$ be closed classes of decision tables from $\mathcal{M%
}_{k}^{\infty }$. Then $A_{1}\cup A_{2}$ is a closed class of decision
tables from $\mathcal{M}_{k}^{\infty }$.

\subsection{Deterministic and Nondeterministic Decision Trees\label{7S1.2}}

A \emph{finite tree with root} is a finite directed tree in which exactly
one node called the \emph{root} has no entering edges. The nodes without
leaving edges are called \emph{terminal} nodes.

\begin{definition}
A $k$-\emph{decision tree} is a finite tree with root, which has at least
two nodes and in which

\begin{itemize}
\item The root and edges leaving the root are not labeled.

\item Each terminal node is labeled with a decision from the set $\omega $.

\item Each node, which is neither the root nor a terminal node, is labeled
with an attribute from the set $P$. Each edge leaving such node is labeled
with a number from the set $E_{k}$.
\end{itemize}
\end{definition}

\begin{ex}
Figures \ref{7fig3} and \ref{7fig4} show $2$-decision trees.
\end{ex}

\begin{figure}[tbp]
\centering
\includegraphics[width=0.305\columnwidth]{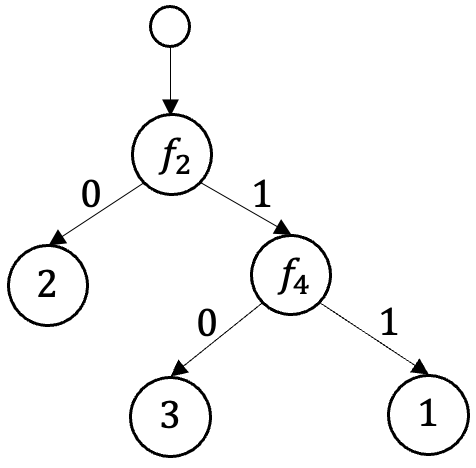}
\caption{A deterministic decision tree for the decision table shown in
Figure   \protect\ref{7fig1}}
\label{7fig3}
\end{figure}

\begin{figure}[tbp]
\centering
\includegraphics[width=0.305\columnwidth]{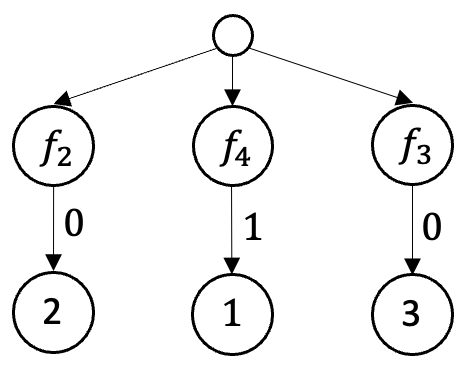}
\caption{A nondeterministic decision tree for the decision table
shown in Figure   \protect\ref{7fig1}}
\label{7fig4}
\end{figure}

We denote by $\mathcal{T}_{k}$ the set of all $k$-decision trees. Let $%
\Gamma \in \mathcal{T}_{k}$. We denote by $\operatorname{At}(\Gamma )$ the set of attributes
attached to nodes of $\Gamma $ that are neither the root nor terminal nodes.
A \emph{complete path} of $\Gamma $ is a sequence $\tau =v_{1},d_{1},\ldots
,v_{m},d_{m},v_{m+1}$ of nodes and edges of $\Gamma $ in which $v_{1}$ is
the root of $\Gamma $, $v_{m+1}$ is a terminal node of $\Gamma $ and, for $%
j=1,\ldots ,m$, the edge $d_{j}$ leaves the node $v_{j}$ and enters the node
$v_{j+1}$. Let $T\in $ $\mathcal{M}_{k}^{\infty }$. If $\operatorname{At}(\Gamma )\subseteq
\operatorname{At}(T)$, then we correspond to the table $T$ and the complete path $\tau $ a
word $\pi (\tau )\in \Omega _{k}(T)$.  If $m=1$, then $\pi (\tau )=\lambda $%
. If $m>1$ and, for $j=2,\ldots ,m$, the node $v_{j}$ is labeled with the
attribute $f_{i_{j}}$ and the edge $d_{j}$ is labeled with the number $%
\delta _{j}$, then $\pi (\tau )=(f_{i_{2}},\delta _{2})\cdots
(f_{i_{m}},\delta _{m})$. Denote $T(\tau )=T\pi (\tau )$.

\begin{definition}
Let $T\in \mathcal{M}_{k}^{\infty }\setminus \{\Lambda \}$. A
\emph{deterministic decision tree for the table} $T$ is a $k$-decision tree $%
\Gamma $ satisfying the following conditions:
\index{Decision table with many-valued decisions!decision tree!deterministic}

\begin{itemize}
\item Only one edge leaves the root of $\Gamma $.

\item For any node, which is neither the root nor a terminal node, edges
leaving this node are labeled with pairwise different numbers.

\item $\operatorname{At}(\Gamma )\subseteq \operatorname{At}(T)$.

\item For any row of $T$, there exists a complete path $\tau $ of $\Gamma $
such that the considered row belongs to the table $T(\tau )$.

\item For any complete path $\tau $ of $\Gamma $, either $T(\tau )=\Lambda $
or the decision attached to the terminal node of $\tau $ is a common
decision for the table $T(\tau )$ .
\end{itemize}
\end{definition}

\begin{ex}
The $2$-decision tree shown in Figure   \ref{7fig3} is a deterministic decision
tree for the decision table shown in Figure   \ref{7fig1}.
\end{ex}

\begin{definition}
Let $T\in $ $\mathcal{M}_{k}^{\infty }\setminus \{\Lambda \}$. A
\emph{nondeterministic decision tree for the table} $T$ is a $k$-decision tree $%
\Gamma $ satisfying the following conditions:
\index{Decision table with many-valued decisions!decision tree!nondeterministic}

\begin{itemize}
\item $\operatorname{At}(\Gamma )\subseteq \operatorname{At}(T)$.

\item For any row of $T$, there exists a complete path $\tau $ of $\Gamma $
such that the considered row belongs to the table $T(\tau )$.

\item For any complete path $\tau $ of $\Gamma $, either $T(\tau )=\Lambda $
or the decision attached to the terminal node of $\tau $ is a common
decision for the table $T(\tau )$.
\end{itemize}
\end{definition}

\begin{ex}
The $2$-decision tree shown in Figure   \ref{7fig4} is a
nondeterministic decision tree for the decision table shown in Figure   \ref%
{7fig1}.
\end{ex}

\subsection{Complexity Measures\label{7S1.3}}

Denote by $P^{*}$ the set of all finite words over the alphabet $P=\{f_{i}:i\in
\omega \}$, which contains the empty word $\lambda $.

\begin{definition}
A \emph{partially bounded complexity measure} is an arbitrary function $\psi :P^{*}\rightarrow
\omega $ that has the following properties: for any words $\alpha
_{1},\alpha _{2}\in P^{*}$,
\index{Complexity measure!partially bounded}

\begin{itemize}
\item $\psi (\alpha _{1})=0$ if and only if $\alpha _{1}=\lambda $ -- \emph{%
positivity} property.

\item $\psi (\alpha _{1})=\psi (\alpha _{1}^{\prime })$ for any word $\alpha
_{1}^{\prime }$ obtained from $\alpha _{1}$ by permutation of letters --
\emph{commutativity} property.

\item $\psi (\alpha _{1})\leq \psi (\alpha _{1}\alpha _{2})$ -- \emph{%
nondecreasing} property.

\item $\psi (\alpha _{1}\alpha _{2})\leq \psi (\alpha _{1})+\psi (\alpha
_{2})$ -- \emph{boundedness from above} property.
\end{itemize}
\end{definition}

The following functions are partially bounded complexity measures:

\begin{itemize}
\item Function $h$ for which, for any word $\alpha \in P^{*}$, $h(\alpha
)=\left\vert \alpha \right\vert $, where $\left\vert \alpha \right\vert $ is
the length of the word $\alpha $. This function is called the \emph{depth}.

\item An arbitrary function $\varphi :P^{*}\rightarrow \omega $ such that $%
\varphi (\lambda )=0$, for any $f_{i}\in P$, $\varphi (f_{i})>0$ and, for
any nonempty word $f_{i_{1}}\cdots f_{i_{m}}\in P^{*}$,
\begin{equation}
\varphi (f_{i_{1}}\cdots f_{i_{m}})=\sum_{j=1}^{m}\varphi (f_{i_{j}}).
\label{7E1}
\end{equation}

This function is called the \emph{weighted depth}.

\item An arbitrary function $\rho :P^{*}\rightarrow \omega $ such that $\rho
(\lambda )=0$, for any $f_{i}\in P$, $\rho (f_{i})>0$, and, for any nonempty
word $f_{i_{1}}\cdots f_{i_{m}}\in P^{*}$, $\rho (f_{i_{1}}\cdots
f_{i_{m}})=\max \{\rho (f_{i_{j}}):j=1,\ldots ,m\}$.
\end{itemize}

\begin{definition}
A \emph{bounded} complexity measure is a partially bounded complexity measure $\psi $, which
has the \emph{boundedness from below} property: for any word $\alpha \in P^{*}$,
$\psi (\alpha )\geq \left\vert \alpha \right\vert $.
\end{definition}
\index{Complexity measure!bounded}

Any partially bounded complexity measure satisfying the equality (\ref{7E1}), in particular the
function $h$, is a bounded complexity measure. One can show that if
functions $\psi _{1}$ and $\psi _{2}$ are partially bounded complexity measures, then the
functions $\psi _{3}$ and $\psi _{4}$ are partially bounded complexity measures, where for any
$\alpha \in P^{*}$, $\psi _{3}(\alpha )=\psi _{1}(\alpha )+\psi _{2}(\alpha )$
and $\psi _{4}(\alpha )=\max (\psi _{1}(\alpha ),\psi _{2}(\alpha ))$. If
the function $\psi _{1}$ is a bounded complexity measure, then the functions
$\psi _{3}$ and $\psi _{4}$ are bounded complexity measures.

\begin{definition}
Let $\psi $ be a partially bounded complexity measure. We extend it to the set of all finite
subsets of the set $P$. Let $D$ be a finite subset of the set $P$. If $%
D=\emptyset $, then $\psi (D)=0$. Let $D=\{f_{i_{1}},\ldots ,f_{i_{m}}\}$
and $m\geq 1$. Then $\psi (D)=\psi (f_{i_{1}}\cdots f_{i_{m}})$.
\end{definition}

\begin{definition}
Let $\psi $ be a partially bounded complexity measure. We extend it to the set of  finite
words $\Omega $ over the alphabet $\{(f_{i},\delta ):f_{i}\in P,\delta \in
\omega \}$ including the empty word $\lambda $. Let $\alpha \in \Omega $. If
$\alpha =\lambda $, then $\psi (\alpha )=0$. Let $\alpha =(f_{i_{1}},\delta
_{1})\cdots (f_{i_{m}},\delta _{m})$ and $m\geq 1$. Then $\psi (\alpha)=\psi
(f_{i_{1}}\cdots f_{i_{m}})$.
\end{definition}

\subsection{Parameters of Decision Trees and Tables\label{7S1.4}}

\begin{definition}
Let $\psi $ be a partially bounded complexity measure. We extend the function $\psi $ to the
set $\mathcal{T}_{k}$. Let $\Gamma \in \mathcal{T}_{k}$. Then $\psi (\Gamma
)=\max \{\psi (\pi (\tau ))\}$, where the maximum is taken over all complete
paths $\tau $ of the decision tree $\Gamma $. For a given partially bounded complexity measure
$\psi $, the value $\psi (\Gamma )$ will be called the \emph{complexity of
the decision tree} $\Gamma $. The value $h(\Gamma )$ will be called the
\emph{depth of the decision tree} $\Gamma $.
\end{definition}

Let $\psi $ be a partially bounded complexity measure. We now describe the functions $\psi ^{d}
$, $\psi ^{a}$, $m_{\psi }$, $W_{\psi }$,  $N$,
 $M_{\psi }$  defined on the set $\mathcal{M}_{k}^{\infty }$  and functions
$Z$, $G$  defined on the set $\mathcal{M}_{2}^{\infty }$ and
taking values from the set $\omega $. By definition, the value of each of
these functions for $\Lambda $ is equal $0$. We also describe the function $%
l_{\psi }$ defined on the set $\mathcal{M}_{k}^{\infty }\times \omega $%
.  By definition, the value of this function for tuple  $(\Lambda ,n)$, $%
n\in \omega $, is equal $0$. Let $T\in \mathcal{M}_{k}^{\infty }\setminus
\{\Lambda \}$ and $n\in \omega $.

\begin{itemize}
\item $\psi ^{d}(T)=\min \{\psi (\Gamma )\}$, where the minimum is taken
over all deterministic decision trees $\Gamma $ for the table $T$.

\item $\psi ^{a}(T)=\min \{\psi (\Gamma )\}$, where the minimum is taken
over all nondeterministic decision trees $\Gamma $ for the table $T$.

\item $m_{\psi }(T)=\max \{\psi (f_{i}):f_{i}\in \operatorname{At}(T)\}$.

\item $W_{\psi }(T)=\psi (\operatorname{At}(T))$.

\item $N(T)$ is the number of rows in the table $T$.

\item If $T\in \mathcal{M}_{k}^{\infty c}$, then $M_{\psi }(T)=0$.
Let $T\notin \mathcal{M}_{k}^{\infty c}$, $\left\vert
\operatorname{At}(T)\right\vert =n$, and columns of the table $T$ be labeled with the
attributes $f_{t_{1}},\ldots ,f_{t_{n}}$. Let $\bar{\delta}=(\delta
_{1},\ldots ,\delta _{n})\in E_{k}^{n}$. Denote by $M_{\psi }(T,\bar{\delta}%
) $ the minimum number $p\in \omega $ for which there exist attributes $%
f_{t_{i_{1}}},\ldots ,f_{t_{i_{m}}}\in \operatorname{At}(T)$ such that $T(f_{t_{i_{1}}},%
\delta _{i_{1}})\cdots (f_{t_{i_{m}}},\delta _{i_{m}})\in \mathcal{M}%
_{k}^{\infty c}$ and $\psi (f_{t_{i_{1}}}\cdots f_{t_{i_{m}}})=p$.
Then $M_{\psi }(T)=\max \{M_{\psi }(T,\bar{\delta}):\bar{\delta}\in
E_{k}^{n}\}$.

\item Let $T\in \mathcal{M}_{2}^{\infty }$. A decision table $Q\in \mathcal{M%
}_{2}^{\infty }$ will be called \emph{complete} if $N(Q)=2^{|\operatorname{At}(Q)|}$. Then $%
Z(T)$ is the maximum number of columns in a complete table from $[T]$ if such tables exist and $0$ otherwise.
\index{Decision table with many-valued decisions!complete}

\item Let $T\in \mathcal{M}_{2}^{\infty }$. A word $\alpha \in \Omega _{2}(T)
$ will be called \emph{annihilating} word for the table $T$ if $T\alpha
=\Lambda $ and $\alpha$ does not contain letters $(f_i,\delta)$ and $(f_i,\sigma)$ such that $\delta \neq \sigma$. An annihilating word $\alpha $ for the table $T$ will be called
\emph{irreducible} if any subword of $\alpha $ obtained from $\alpha $ by
removal of some letters and different from $\alpha $ is not annihilating.
Then $G(T)$ is the maximum length of an irreducible annihilating word for
the table $T$ if such words exist and $0$ otherwise.
\index{Decision table with many-valued decisions!annihilating word}
\index{Decision table with many-valued decisions!annihilating word!irreducible}

\item Denote $\Omega
_{k}^{n}(T)=\{\alpha :\alpha \in \Omega _{k}(T),\psi (\alpha )\leq n\}$. A
finite set $U\subseteq \Omega _{k}^{n}(T)$ will be called $(\psi ,n)$-\emph{%
cover }of the table $T$ if $\bigcup_{\alpha \in U}\Delta (T\alpha
)=\Delta (T)$. The $(\psi ,n)$-cover $U$ will be called \emph{irreducible}
if each proper subset of $U$ is not a $(\psi ,n)$-cover of the table
$T$. Then $l_{\psi }(T,n)$ is the maximum cardinality of an irreducible $%
(\psi ,n)$-cover of the table $T$. It is clear that $\{\lambda\}$ is an irreducible $(\psi ,n)$-cover of the table $T$. Therefore $l_{\psi }(T,n) \ge 1$.
\index{Decision table with many-valued decisions!$(\psi ,n)$-cover}
\index{Decision table with many-valued decisions!$(\psi ,n)$-cover!irreducible}

\end{itemize}

For the bounded complexity measure $h$, we denote $W(T)=W_{h}(T)$,  $M(T,\bar{\delta})=M_{h}(T,\bar{\delta})$ and $M(T)=M_{h}(T)$. Note that $W(T)$ is the
number of columns in the table $T$.

\begin{ex}
We denote by $T_{0}$ the decision table shown in Figure   \ref{7fig1}. One can
show that $h^{d}(T_{0})=2$, $h^{a}(T_{0})=1$, $m_h(T_0)=1$, $%
W(T_{0})=3 $,  $N(T_{0})=6$, $M(T_0)=2$, $Z(T_{0})=2$, $G(T_{0})=3$,
$l_{h}(T_{0}, 0)=1$,  $l_{h}(T_{0}, 1)=3$ and $l_{h}(T_{0},n)=6$ for any
$n\in \omega\setminus \{0,1\}$.
\end{ex}

\section{Main Results\label{7S2}}

In this section, we consider results obtained for the function $\mathcal{H}^{\infty}_{\psi ,A}$ and
discuss closed classes of decision tables generated by information systems.

\subsection{Function $\mathcal{H}^{\infty}_{\protect\psi ,A}$\label{7S2.3}}

Let $\psi $ be a bounded complexity measure and $A$ be a nontrivial closed
class of decision tables from $\mathcal{M}_{k}^{\infty }$. We now define
possibly partial function $\mathcal{H}^{\infty}_{\psi ,A}:\omega \rightarrow \omega $%
. Let $n\in \omega $. If the set $\{\psi ^{d}(T):T\in A,\psi ^{a}(T)\leq n\}$
is infinite, then the value $\mathcal{H}^{\infty}_{\psi ,A}(n)$ is undefined.
Otherwise, $\mathcal{H}^{\infty}_{\psi ,A}(n)=\max \{\psi ^{d}(T):T\in A,\psi
^{a}(T)\leq n\}$.

The function $\mathcal{H}^{\infty}_{\psi ,A}$ characterizes the growth in the worst
case of the minimum complexity of deterministic decision trees for decision
tables from $A$ with the growth of the minimum complexity of
nondeterministic decision trees for these tables.

We now define possibly partial function $L_{\psi ,A}:\omega \rightarrow
\omega $. Let $n\in \omega $. If the set $\{l_{\psi }(T,n):T\in A\}$ is
infinite, then the value $L_{\psi ,A}(n)\ $is not defined. Otherwise, $%
L_{\psi ,A}(n)=\max \{l_{\psi }(T,n):T\in A\}$. One can show that  in this case $L_{\psi ,A}(n) \ge 1$.

The following statement describes a criterion for the function $\mathcal{H}^{\infty}%
_{\psi ,A}$ to be everywhere defined.

\begin{theorem}
\label{7T3} Let $\psi $ be a bounded complexity measure and $A$ be a nontrivial
closed class of decision tables from $\mathcal{M}_{k}^{\infty }$. Then the
function $\mathcal{H}^{\infty}_{\psi ,A}$ is everywhere defined if and only if the
function $L_{\psi ,A}$ is everywhere defined.
\end{theorem}

We now describe two possible types of  behavior for everywhere defined function $\mathcal{H}_{\psi ,A}$.

Let $D=\{n_{i}:i\in \omega \}$ be an infinite subset of the set $\omega $ in
which, for any $i\in \omega $, $n_{i}<n_{i+1}$. We now define a function $%
H_{D}:\omega \rightarrow \omega $. Let $n\in \omega $. If $n<n_{0}$, then $%
H_{D}(n)=0$. If, for some $i\in \omega $, $n_{i}\leq n<n_{i+1}$, then $%
H_{D}(n)=n_{i}$.

\begin{theorem}
\label{7T4} Let $\psi $ be a bounded complexity measure, $A$ be a nontrivial
closed class of decision tables from $\mathcal{M}_{k}^{\infty }$, and the
function $\mathcal{H}^{\infty}_{\psi ,A}$ be everywhere defined. Then $\mathcal{H}^{\infty}_{\psi ,A}$ is a
nondecreasing function and $\mathcal{H}^{\infty}_{\psi ,A}(0)=0$. For this function, one of the following statements
holds:

{\rm (a)} If the function $\psi ^{d}$ is bounded from above on the class $A$, then
there is a nonnegative constant $c$ such that $\mathcal{H}^{\infty}_{\psi ,A}(n)\leq c$ for any
$n\in \omega $.

{\rm (b)} If the function $\psi ^{d}$ is not bounded from above on the class $A$,
then there exists an infinite subset $D$ of the set $\omega $ such that $%
\mathcal{H}^{\infty}_{\psi ,A}(n)\geq H_{D}(n)$ for any $n\in \omega $.
\end{theorem}

The following statement shows a wide spectrum of behavior of the
everywhere defined function $\mathcal{H}^{\infty}_{\psi ,A}$ in the case when the function $\psi^{d}$ is not bounded from above on the class $A$.

\begin{theorem}
\label{7T5} Let $\varphi :\omega \rightarrow \omega $ be a nondecreasing
function such that $\varphi (n)\geq n$ for any $n\in \omega $ and $\varphi
(0)=0$. Then there exist a closed class $A$ of decision tables from $%
\mathcal{M}_{k}^{\infty }$ and a bounded complexity measure $\psi $ such
that the function $\mathcal{H}^{\infty}_{\psi ,A}$ is everywhere defined and $\varphi (n)\leq \mathcal{H}^{\infty}_{\psi ,A}(n)\leq \varphi (n)+n$ for any $n\in \omega
$.
\end{theorem}

Let $A$ be a nontrivial closed class of decision tables from $\mathcal{M}%
_{2}^{\infty }$ and $\psi $ be a bounded complexity measure.
Let $n\in \omega $. Denote $A_{\psi }(n)=\{T:T\in A,m_{\psi }(T)\leq n\}$.
Since $\Lambda \in A$, both sets $A_{\psi }(n)$ and $\{Z(T):T\in A_{\psi }(n)\}$ are
nonempty sets. We now define probably partial function $Z_{\psi ,A}:\omega
\rightarrow \omega $. If $\{Z(T):T\in A_{\psi }(n)\}$ is an infinite set, then the
value $Z_{\psi ,A}(n)$ is not defined. Otherwise, $Z_{\psi ,A}(n)=\max
\{Z(T):T\in A_{\psi }(n)\}$.

Let $n\in \omega $. Since $\Lambda \in A$, the set $\{G(T):T\in A_{\psi }(n)\}$ is
nonempty. We now define probably partial function $G_{\psi ,A}:\omega
\rightarrow \omega $. If $\{G(T):T\in A_{\psi }(n)\}$ is an infinite set,
then the value $G_{\psi ,A}(n)$ is not defined. Otherwise, $G_{\psi
,A}(n)=\max \{G(T):T\in A_{\psi }(n)\}$.

The following statement describes the criterion for the everywhere defined
function $\mathcal{H}^{\infty}_{\psi ,A}$ to be bounded from above by a polynomial. For simplicity, we consider here only decision tables from the set $\mathcal{M}%
_{2}^{\infty }$.

\begin{theorem}
\label{7T6} Let $A$ be a nontrivial closed class of decision tables from $\mathcal{M}%
_{2}^{\infty }$, $\psi $ be a bounded complexity measure and the function $%
\mathcal{H}^{\infty}_{\psi ,A}$ be everywhere defined. Then a polynomial $p_{0}$ such that $%
\mathcal{H}^{\infty}_{\psi ,A}(n)\leq p_{0}(n)$ for any $n\in \omega $ exists if and only if
there exist polynomials $p_{1}$, $p_{2}$, and $p_{3}$ such that $Z_{\psi
,A}(n)\leq p_{1}(n)$, $G_{\psi ,A}(n)\leq p_{2}(n)$ and $L_{\psi ,A}(n)\leq
2^{p_{3}(n)}$ for any $n\in \omega $.
\end{theorem}

\subsection{Family of Closed Classes of Decision Tables\label{7S2.4}}

Let $U$ be a set and $\Phi =\{f_{0},f_{1},\ldots \}$ be a finite or
countable set of functions (attributes) defined on $U$ and taking values
from $E_{k}$. The pair $(U,\Phi )$ is called a $k$-\emph{information system}%
. A \emph{problem} over $(U,\Phi )$ is an arbitrary tuple $z=(U,\nu
,f_{i_{1}},\ldots ,$ $f_{i_{n}})$, where $n\in \omega \setminus \{0\}$, $\nu
:E_{k}^{n}\rightarrow \mathcal{P}(\omega )$ and $f_{i_{1}},\ldots ,f_{i_{n}}$
are functions from $\Phi $ with pairwise different indices $i_{1},\ldots
,i_{n}$. The problem $z$ is to determine a value from the set $\nu
(f_{i_{1}}(u),\ldots ,f_{i_{n}}(u))$ for a given $u\in U$. Various examples
of $k$-information systems can be found in
\cite{Moshkov05}.

We denote by $T(z)$ a decision table from $\mathcal{M}_{k}^{\infty }$ with $%
n $ columns labeled with attributes $f_{i_{1}},\ldots ,f_{i_{n}}$. A row $%
(\delta _{1},\ldots ,\delta _{n})\in E_{k}^{n}$ belongs to the table $T(z)$
if and only if the system of equations $\{f_{i_{1}}(x)=\delta _{1},\ldots
,f_{i_{n}}(x)=\delta _{n}\}$ has a solution from the set $U$. This row is
labeled with the set of decisions $\nu (\delta _{1},\ldots ,\delta _{n})$%
.

Let the algorithms for the problem $z$ solving be algorithms in which each
elementary operation consists in calculating the value of some attribute
from the set $\{f_{i_{1}},\ldots ,f_{i_{n}}\}$ on a given element $u\in U$.
Then, as a model of the problem $z$, we can use the decision table $T(z)$,
and as models of algorithms for the problem $z$ solving -- deterministic and
nondeterministic decision trees for the table $T(z)$.

Denote by $\operatorname{Probl}^{\infty }(U,\Phi )$ the set of problems over $(U,\Phi
)$ and $\operatorname{Tab}^{\infty }(U,\Phi )=\{T(z):z\in \operatorname{Probl}^{\infty
}(U,\Phi )\}$. One can show that $\operatorname{Tab}^{\infty }(U,\Phi )=[\operatorname{Tab}%
^{\infty }(U,\Phi )]$, i.e., $\operatorname{Tab}^{\infty }(U,\Phi )$ is a closed
class of decision tables from $\mathcal{M}_{k}^{\infty }$ \emph{generated}
by the information system $(U,\Phi )$.
\index{Decision table with many-valued decisions!closed class!generated by information system}

Closed classes of decision tables generated by $k$-information systems are
the most natural examples of closed classes. However, the notion of a closed
class is essentially wider. In particular, the union $\operatorname{Tab}^{\infty
}(U_{1},\Phi _{1})\cup \operatorname{Tab}^{\infty }(U_{2},\Phi _{2})$, where $%
(U_{1},\Phi _{1})$ and $(U_{2},\Phi _{2})$ are $k$-information systems, is a
closed class, but generally, we cannot find an information system $(U,\Phi )$
such that $\operatorname{Tab}^{\infty }(U,\Phi )=\operatorname{Tab}^{\infty }(U_{1},\Phi
_{1})\cup \operatorname{Tab}^{\infty }(U_{2},\Phi _{2})$.

\subsection{Example of Information System\label{7S2.5}}

Let $%
\mathbb{R}
$ be the set of real numbers and $F=\{f_{i}:i\in \omega \}$ be the set of
functions defined on $%
\mathbb{R}
$ and taking values from the set $E_{2}$ such that, for any $i\in \omega $
and $a\in
\mathbb{R}
$,%
\begin{equation*}
f_{i}(a)=\left\{
\begin{array}{ll}
0, & a<i, \\
1, & a\geq i.%
\end{array}%
\right.
\end{equation*}

Let $\psi $ be a bounded complexity measure and $A=\operatorname{Tab}^{\infty }(%
\mathbb{R}
,F)$. One can prove the following statements:

\begin{itemize}
\item The function $\psi^d$ is not bounded from above on the set $A$.

\item The function $L_{\psi ,A}$ is everywhere defined if and only if, for
any $n\in \omega $, the set $\{i:i\in \omega ,\psi (f_{i})\leq n\}$ is
finite.

\item A polynomial $p$ such that $L_{\psi ,A}(n)\leq 2^{p(n)}$ for any $n\in
\omega $ exists if and only if there exists a polynomial $q$ such that $%
\left\vert \{i:i\in \omega ,\psi (f_{i})\leq n\}\right\vert \leq 2^{q(n)}$
for any $n\in \omega $.

\item For any $n\in \omega $, $Z_{\psi ,A}(n)\leq 1$.

\item For any $n\in \omega $, $G_{\psi ,A}(n)\leq 2$.
\end{itemize}

\section{Proofs of Theorems \protect\ref{7T3}, \protect\ref{7T4} and \protect
\ref{7T5}\label{7S4}}

First, we consider a number of auxiliary statements.
It is not difficult to prove the following upper bound on the minimum
complexity of deterministic decision trees for a decision table.

\begin{lemma}
\label{7L1}For any partially bounded complexity measure $\psi $ and any table $T$ from $%
\mathcal{M}_{k}^{\infty }$,
\begin{equation*}
\psi ^{d}(T)\leq W_{\psi }(T).
\end{equation*}
\end{lemma}

The next upper bound on the minimum complexity of deterministic decision
trees for a decision table follows from Corollary 5.2 from \cite{Moshkov20}.

\begin{lemma}
\label{7L2} For any partially bounded complexity measure $\psi $ and any table $T$ from $%
\mathcal{M}_{k}^{\infty }$,
\begin{equation*}
\psi ^{d}(T)\leq \left\{
\begin{array}{ll}
0, & T\in \mathcal{M}_{k}^{\infty c}, \\
M_{\psi }(T)\log _{2}N(T), & T\notin \mathcal{M}_{k}^{\infty c}.%
\end{array}%
\right.
\end{equation*}
\end{lemma}

\begin{lemma}
\label{7L6} For any partially bounded complexity measure $\psi $ and any table $T$ from $%
\mathcal{M}_{k}^{\infty }$,
\begin{equation*}
\psi ^{a}(T)\leq \psi ^{d}(T).
\end{equation*}
\end{lemma}

\begin{proof}
Let $T\in \mathcal{M}_{k}^{\infty }$. If $T=\Lambda $, then $\psi
^{a}(T)=\psi ^{d}(T)=0$. Let $T\in \mathcal{M}_{k}^{\infty }\setminus
\{\Lambda \}$. It is clear that each deterministic decision tree for the
table $T$ is a nondeterministic decision tree for the table $T$. Therefore $%
\psi ^{a}(T)\leq \psi ^{d}(T)$.
\end{proof}

\begin{lemma}
\label{7M1}Let $A$ be a nontrivial closed class of decision tables from $\mathcal{M}%
_{k}^{\infty }$, $\psi $ be a bounded complexity measure, $T\in A$, $n\in
\omega $ and $l_{\psi }(T,n)>0$. Then there exists a mapping $\nu
:E_{k}^{W(T)}\rightarrow \mathcal{P}(\omega )$ such that, for the table $%
T^{\ast }=J(\nu ,T)$, $\psi ^{a}(T^{\ast })\leq n$ and $\psi ^{d}(T^{\ast
})\geq \log _{k}l_{\psi }(T,n)$.
\end{lemma}

\begin{proof}
Let $U$ be an irreducible $(\psi ,n)$-cover of the table $T$ such
that $\left\vert U\right\vert =l_{\psi }(T,n)$. Let $U=\{\alpha _{1},\ldots
,\alpha _{l}\}$, where $l=l_{\psi }(T,n)$. We now define a mapping $\nu
:E_{k}^{W(T)}\rightarrow \mathcal{P}(\omega )$. Let $\bar{\delta}\in
E_{k}^{W(T)}$. If $\bar{\delta}\notin \Delta (T)$, then $\nu (\bar{\delta}%
)=\{0\}$. Let $\bar{\delta}\in \Delta (T)$. Then $\nu (\bar{\delta}%
)\subseteq \{1,\ldots ,l\}$. For any $j\in \{1,\ldots ,l\}$, $j\in \nu (\bar{%
\delta})$ if and only if $\bar{\delta}\in \Delta (T\alpha _{j})$.
Since the
cover $U$ is an irreducible $(\psi ,n)$-cover of the table $%
T$, for each $j\in \{1,\ldots ,l\}$, there is a row $\bar{\delta}_{j}$
of the table $T$ such that $\nu (\bar{\delta}_{j})=\{j\}$. Set $T^{\ast }=J(\nu ,T)$. Let $\Gamma $ be a
deterministic decision tree for the table $T^{\ast }$ such that $h(\Gamma
)=h^{d}(T^{\ast })$. Denote by $L_{t}(\Gamma )$ the number of terminal nodes
of $\Gamma $. Evidently, $l\leq L_{t}(\Gamma )$. One can show that $%
L_{t}(\Gamma )\leq k^{h(\Gamma )}$. Therefore $0<l\leq k^{h(\Gamma )}$.
Hence $h(\Gamma )\geq \log _{k}l$. Taking into account that $h(\Gamma
)=h^{d}(T^{\ast })$, we obtain $h^{d}(T^{\ast })\geq \log _{k}l_{\psi }(T,n)$%
. Using boundedness from below property of the function $\psi $, we obtain $%
\psi ^{d}(T^{\ast })\geq \log _{k}l_{\psi }(T,n)$. It is easy to construct a
nondeterministic decision tree $G$ for the table $T^{\ast }$ with $l=l_{\psi
}(T,n)$ complete paths $\tau _{1},\ldots ,\tau _{l}$ such that
$\pi(\tau _{1})=\alpha_1,\ldots ,\pi(\tau _{l})=\alpha_l$. It is clear that $\psi (G)\leq n$. Therefore $%
\psi ^{a}(T^{\ast })\leq n$.
\end{proof}

\begin{proof}[Proof of Theorem \protect\ref{7T3}]
Let the function $L_{\psi ,A}$ be everywhere defined. Let $T\in A$ and $\psi
^{a}(T)\leq n$. Let $\Gamma $ be a nondeterministic decision tree for the
table $T$ such that $\psi (\Gamma )=\psi ^{a}(T)$ and $U$ be the set of
words $\pi(\tau)$ corresponding to the complete paths $\tau$ in $\Gamma $. Then $U$ is a $(\psi
,n)$-cover of the table $T$. Let $U^{\prime }$ be a subset of the
set $U$, which is an irreducible $(\psi ,n)$-cover of the table $T$.
Then $\left\vert U^{\prime }\right\vert \leq L_{\psi ,A}(n)$.
Let $U^{\prime }=\{\alpha_1, \ldots , \alpha_t\}$.
We now
consider a deterministic decision tree $G$ for the table $T$, which, for a given $\bar{\delta}\in \Delta (T)$,
sequentially iterates through the words $\alpha_i$ from $U^{\prime }$ and check if $\bar{\delta}\in \Delta (T\alpha_i )$ by computing values of attributes from $\alpha_i$. For the first $i$ such that $\bar{\delta}\in \Delta (T\alpha_i )$, the decision tree $G$ returns a common decision for the table $T\alpha_i
$. Using boundedness from above property of the function $\psi $, we obtain
that $\psi (G)\leq nL_{\psi ,A}(n)$. Therefore $\mathcal{H}^{\infty}_{\psi ,A}(n)\leq nL_{\psi
,A}(n)$.

Let the function $L_{\psi ,A}$ be not everywhere defined. Then there exist
a number $n\in \omega $ and an infinite sequence of tables $%
T_{0},T_{1},\ldots $ from $A$ such that $l_{\psi }(T_{0},n)<l_{\psi
}(T_{1},n)<\cdots $. Let $i\in \omega $. Using Lemma \ref{7M1}, we obtain
that there exists a mapping $\nu :E_{k}^{W(T_{i})}\rightarrow \mathcal{P}%
(\omega )$ such that, for the table $T_{i}^{\ast }=J(\nu ,T_{i})$, $\psi
^{a}(T_{i}^{\ast })\leq n$ and $\psi ^{d}(T_{i}^{\ast })\geq \log
_{k}l_{\psi }(T_{i},n)$. Evidently, $T_{i}^{\ast } \in A$. Therefore, the value $\mathcal{H}^{\infty}_{\psi ,A}(n)$ is not
defined.
\end{proof}

\begin{proof}[Proof of Theorem \protect\ref{7T4}]
Immediately from the definition of the function $\mathcal{H}^{\infty}_{\psi ,A}$ it follows
that $\mathcal{H}^{\infty}_{\psi ,A}(n)\leq \mathcal{H}^{\infty}_{\psi ,A}(n+1)$ for any $n\in \omega $. Let $T\in A$ and $\psi ^a(T)\leq 0$. Using
positivity property of the function $\psi $, one can show that $T \in \mathcal{M}_{k}^{\infty c}$. From Lemma \ref{7L2} it follows that $\psi ^d(T)= 0$.
Therefore $\mathcal{H}^{\infty}_{\psi ,A}(0)=0$.

(a) Let there exist a nonnegative constant $c\ $such that $\psi ^{d}(T)\leq
c $ for any table $T\in A$. Then, evidently, $\mathcal{H}^{\infty}_{\psi ,A}(n)\leq c$ for any $%
n\in \omega $.

(b) Let there be no a nonnegative constant $c$ such that $\psi ^{d}(T)\leq c$
for any table $T\in A$. Let us assume that there exists a nonnegative
constant $d\ $such that $\psi ^{a}(T)\leq d$ for any table $T\in A$. Then
the value $\mathcal{H}^{\infty}_{\psi ,A}(d)$ is not defined but this is impossible. Therefore
the set $D=\{\psi ^{a}(T):T\in A\}$ is infinite. Since $\Lambda \in A$, $0 \in D$.
By Lemma \ref{7L6}, $\psi
^{a}(T)\leq \psi ^{d}(T)$ for any table $T\in A$. Hence $\mathcal{H}^{\infty}_{\psi ,A}(n)\geq
n $ for any $n\in D$. Taking into account that $\mathcal{H}^{\infty}_{\psi ,A}$ is a
nondecreasing function, we obtain that $\mathcal{H}^{\infty}_{\psi ,A}(n)\geq H_{D}(n)$ for any
$n\in \omega $.
\end{proof}

\begin{proof}[Proof of Theorem \protect\ref{7T5}]
For each $n\in \omega \setminus \{0\}$, we define a decision table $%
Q_{n} $.  The
union of closures of these tables forms a closed class $A$. In the same
time, we define a weighted depth $\psi $.

Let $n\in \omega \setminus \{0\}$ and $\varphi (n)=nm+j$, where $0\leq j\leq
n-1$. Denote by $Q_{n}$ a decision table from $\mathcal{M}_{2}^{\infty }$
with $m+2$ columns labeled with pairwise different attributes $%
f_{i(1,n)},\ldots ,f_{i(m+2,n)}$, respectively, and $m+2$ rows labeled
with the sets of decisions $\{1\},\ldots ,\{m+2\}$, respectively. The table $%
Q_{n}$ is filled with zeros with the exception of the main diagonal that is
filled with ones. If $j=0$, then we should remove from the table $Q_{n}$ the
first row and the first column. If $j>0$, then $\psi (f_{i(1,n)})=j$. For
$t=2,\ldots ,m+2$, $\psi (f_{i(t,n)})=n$. We choose attributes such that  $\operatorname{At}(Q_{l})\cap \operatorname{At}(Q_{t})=\emptyset $ if $l\neq t$. If $f_{i}\in P\setminus \bigcup_{n\in
\omega \setminus \{0\}}\operatorname{At}(Q_{n})$, then $\psi (f_{i})=1$.
Denote $A=[\{Q_{1},Q_{2},\ldots
\}] $.

One can show that, for any $n\in \omega $, $L_{\psi ,A}(n)\leq \varphi (n)+2$%
, i.e., the function $L_{\psi ,A}$ is everywhere defined. By Theorem \ref{7T3}%
, the function $\mathcal{H}^{\infty}_{\psi ,A}$ is everywhere defined. By Theorem \ref{7T4}, $\mathcal{H}^{\infty}_{\psi
,A}(0)=0$. Let $n\in \omega \setminus \{0\}$. One can show that $\psi
^{a}(Q_{n})=n$ and $\psi ^{d}(Q_{n})=nm+j=\varphi (n)$. Hence $H_{\psi
,A}(n)\geq \varphi (n)$.

Let us show that, for any $n\in \omega $, $\mathcal{H}^{\infty}_{\psi ,A}(n)\leq \varphi (n)+n$%
. Evidently, this inequality holds if $n=0$. Let $n>0$, $T\in A$ and $%
\psi ^{a}(T)\leq n$. Let $T\in \lbrack \{Q_{1},\ldots ,Q_{n}\}]$. By Lemma \ref{7L1}, $\psi ^{d}(T)\leq W_{\psi }(T)\leq $ $\varphi (n)+n$. Let $T\in
\lbrack \{Q_{n+1},Q_{n+2},\ldots \}]$. Let $\Gamma $ be a nondeterministic
decision tree for the table $T$ such that $\psi (\Gamma )\leq n$. Using
commutativity and nondecreasing properties of the function $\psi $, we obtain that
$\psi (f_{i}) \le n$ for any $f_{i}\in \operatorname{At}(\Gamma)$. It is clear that $\left\vert
\{f_{i}:f_{i}\in \operatorname{At}(T),\psi (f_{i})\leq n\}\right\vert \leq 1$. Using these facts, one can show that $\psi^d(T)\leq n$.
\end{proof}

\section{Proof of Theorem \protect\ref{7T6} \label{7S5}}

We preface the proof of the theorem with several auxiliary statements.

\begin{lemma}
\label{7M2}Let $A$ be a nontrivial closed class of decision tables from $\mathcal{M}%
_{k}^{\infty }$, $\psi $ be a bounded complexity measure, and the function $%
L_{\psi ,A}$ be everywhere defined. Then, for any $n\in \omega $ such that $%
L_{\psi ,A}(n)>0$, $\mathcal{H}^{\infty}_{\psi ,A}(n)\geq \log _{k}L_{\psi ,A}(n)$.
\end{lemma}

\begin{proof}
By Theorem \ref{7T3}, the function $\mathcal{H}^{\infty}_{\psi ,A}$ is everywhere defined.
Let $n\in \omega $ and $L_{\psi ,A}(n)>0$. Let $T\in A$ and $l_{\psi
}(T,n)=L_{\psi ,A}(n)$. Using Lemma \ref{7M1}, we obtain that there exists a
mapping $\nu :E_{k}^{W(T)}\rightarrow \mathcal{P}(\omega )$ such that, for
the table $T^{\ast }=J(\nu ,T)$, $\psi ^{a}(T^{\ast })\leq n$ and $\psi
^{d}(T^{\ast })\geq \log _{k}l_{\psi }(T,n)=\log _{k}L_{\psi ,A}(n)$.
Therefore $\mathcal{H}^{\infty}_{\psi ,A}(n)\geq \log _{k}L_{\psi ,A}(n)$.
\end{proof}

\begin{figure}[tbp]
\centering
\includegraphics[width=0.5\columnwidth]{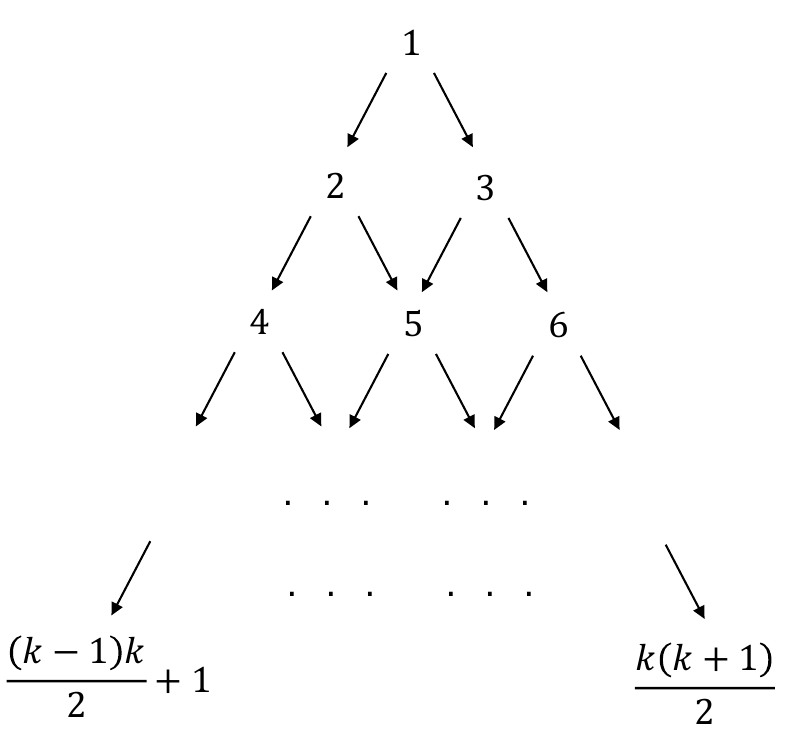}
\caption{Graph $G_{k}$}
\label{7fig5}
\end{figure}

It is not difficult to prove the following statement.

\begin{lemma}
\label{7M3}Let $\psi $ be a bounded complexity measure and $T\in \mathcal{M%
}_{k}^{\infty }\setminus \mathcal{M}_{k}^{\infty c}$. Then $\psi ^{a}(T)=\max \{M_{\psi}(T,\bar{\delta}):\bar{\delta}%
\in \Delta (T)\}$.
\end{lemma}

Let $k\in \omega \setminus \{0\}$. Let us consider the graph $G_{k}$
depicted in Figure   \ref{7fig5}. Nodes of this graph are divided into $k$ layers
numbered from top to bottom. The number of nodes in this graph is equal to $%
m(k)=\frac{k(k+1)}{2}$. Nodes of the graph $G_{k}$ are numbered as it is
shown in Figure   \ref{7fig5}. We denote by $l(i)$ the number of the left-hand child of
the node $i$ and by $p(i)$ the number of the right-hand child of the node $i$%
. We now define mapping $\nu _{k}:E_{2}^{m(k)}\rightarrow \mathcal{P}(\omega
)$. Let $\bar{\delta}=(\delta _{1},\ldots ,\delta _{m(k)})\in E_{2}^{m(k)}$.
Then $\nu _{k}(\bar{\delta})\subseteq \{0,1,\ldots ,m(k)\}$. The number $0$
belongs to the set $\nu _{k}(\bar{\delta})$ if and only if $\delta _{1}=0$.
Let $i\in \{1,\ldots ,m(k)\}$. If $i$ is not a node of the $k$th layer, then
$i\in \nu _{k}(\bar{\delta})$ if and only if $\delta _{i}=1$ and $\delta
_{l(i)}=\delta _{p(i)}=0$. If $i$ is a node of the $k$th layer, then $i\in
\nu _{k}(\bar{\delta})$ if and only if $\delta _{i}=1$.

Let $T_{k}$ be a decision table from $\mathcal{M}_{2}^{\infty }$ with $m(k)$
columns labeled with the attributes $f_{1},\ldots ,f_{m(k)}$ and $2^{m(k)}$
pairwise different rows filled with numbers from $E_{2}$. Any row $(\delta
_{1},\ldots ,$ $\delta _{m(k)})$ of the table $T_{k}$ is labeled with the set
of decisions $\nu _{k}(\delta _{1},\ldots ,\delta _{m(k)})$. It is clear
that $T_{k}$ is a complete decision table.

\begin{lemma}
\label{7M4} For any $k\in \omega \setminus \{0\}$, $h^{a}(T_{k})\leq 3$.
\end{lemma}

\begin{proof}
Let $\bar{\delta}=(\delta _{1},\ldots ,\delta _{m(k)})\in E_{2}^{m(k)}$. If $%
\delta _{1}=0$, then $0\in \Pi (T_{k}(f_{1},0))$, $T_{k}(f_{1},0)\in
\mathcal{M}_{2}^{\infty c}$, and $M(T_{k},\bar{\delta})\leq 1$.
Let $\delta _{1}=1$. If the node $1$ belongs to the $k$th layer, then $1\in \Pi (T_{k}(f_{1},1))$, $%
T_{k}(f_{1},1)\in \mathcal{M}_{2}^{\infty c}$, and $M(T_{k},\bar{%
\delta})\leq 1$. Let the node $1$ do not belong to the $k$th layer. If $%
\delta _{l(1)}=0$ and $%
\delta _{p(1)}=0$, then $1\in \Pi
(T_{k}(f_{1},1)(f_{l(1)},0)(f_{p(1)},0))$, $%
T_{k}(f_{1},1)(f_{l(1)},0)(f_{p(1)},0)\in \mathcal{M}_{2}^{\infty c}$, and $M(T_{k},\bar{\delta})\leq 3$. Let $1\in \{\delta _{l(1)},\delta
_{p(1)}\}$ and, for definiteness, $\delta _{l(1)}=1$. Then, instead of the
number $1$, we consider the number $l(1)$, etc. As a result, we either find $%
i\in \{1,\ldots ,m(k)\}$ such that the node $i$ does not belong to the $k$th
layer of the graph $G_{k}$ and for which $\delta _{i}=1$ and $\delta
_{l(i)}=\delta _{p(i)}=0$, or we find $i\in \{1,\ldots ,m(k)\}$ such that
the node $i$ belongs to the $k$th layer of the graph $G_{k}$ and for which $%
\delta _{i}=1$. In both cases, $M(T_{k},\bar{\delta})\leq 3$. Thus, $M(T_{k},%
\bar{\delta})\leq 3$ for any $\bar{\delta}\in E_{2}^{m(k)}$ and $M(T)\leq 3.$
Using Lemma \ref{7M3}, we obtain that $h^{a}(T_{k})\leq 3$.
\end{proof}

A \emph{complete path} in the graph $G_{k}$ is a directed path from the node
of the $1$st layer to a node of the $k$th layer. We correspond to a complete
path $\pi $ in $G_{k}$ its \emph{characteristic tuple} $\delta (\pi
)=(\delta _{1},\ldots ,\delta _{m(k)})\in E_{2}^{m(k)}$, where, for $%
i=1,\ldots ,m(k)$, $\delta _{i}=1$ if and only if $\pi $ passes through the
node $i$ of $G_{k}$. Denote by $Path_{k}$ the set of complete paths in the
graph $G_{k}$. Evidently $\left\vert Path_{k}\right\vert =2^{k-1}$. Denote by $%
T_{k}^{\ast }$ a subtable of the table $T_{k}$, which contains only rows that
are characteristic tuples of paths from $Path_{k}$. It is clear that each row
from $T_{k}^{\ast }$ is labeled with a singleton set in which the only
element is the number of the last node in the path corresponding to the
considered row. Let $T$ be a subtable of the table $T_{k}^{\ast }$ and $\bar{%
\delta}_{1},\ldots ,\bar{\delta}_{t}$ be all rows of $T$. Denote $%
r(T)=\left\vert \{\nu _{k}(\bar{\delta}_{1}),\ldots ,\nu _{k}(\bar{\delta}%
_{t})\}\right\vert $. Evidently, $T\in \mathcal{M}_{2}^{\infty c}$
if and only if $r(T)\leq 1$.

\begin{lemma}
\label{7M5} Let $T=T_{k}^{\ast}(f_{i_{1}},0)\cdots $ $(f_{i_{n}},0)$, where $k\in \omega \setminus \{0\}$ and $f_{i_{1}},\ldots,f_{i_{n}}\in \{f_{1},\ldots ,f_{m(k)}\}$. If $T\neq \Lambda $, then $r(T)\geq
\max (1,k-n)$.
\end{lemma}

\begin{proof}
We will prove the considered statement by induction on $k$. Let $k=1$. Then,
evidently, the considered statement holds. Let it hold for $k-1$ for some $%
k\geq 2$. Let $f_{i_{1}},\ldots ,f_{i_{n}}\in \{f_{1},\ldots ,f_{m(k)}\}$,
$T=T_{k}^{\ast }(f_{i_{1}},0)\cdots (f_{i_{n}},0)$ and $T\neq \Lambda $%
. Without loss of generality we assume that $i_{1},\ldots ,i_{t}\leq m(k-1)$
and $i_{t+1},\ldots ,i_{n}>m(k-1)$. Since $T\neq \Lambda $, $T_{k-1}^{\ast
}(f_{i_{1}},0)\cdots (f_{i_{t}},0)\neq \Lambda $. By induction hypothesis, $%
r(T_{k-1}^{\ast }(f_{i_{1}},0)\cdots (f_{i_{t}},0))\geq \max (1,k-1-t))$.
Let $j_{1}<\cdots <j_{s}$ be all elements of sets attached to rows of the
table $T_{k-1}^{\ast }(f_{i_{1}},0)\cdots (f_{i_{t}},0)$. Then, evidently, $%
l(j_{1})<p(j_{1})<\cdots <p(j_{s})$ are elements of sets attached to rows of
the table $T_{k}^{\ast }(f_{i_{1}},0)\cdots (f_{i_{t}},0)$. Let us show that
$r(T)\geq \max (1,k-n)$. If $n\geq k$, then this inequality follows from the
fact that $T\neq \Lambda $. Let $n<k$. Then the following two cases are
possible: (i) $t=k-1$ and (ii) $t<k-1$. In the first case, the considered
inequality follows from the fact that $T\neq \Lambda $. In the second case, $%
r(T_{k-1}^{\ast }(f_{i_{1}},0)\cdots (f_{i_{t}},0))\geq k-t-1$ and, as it
was shown earlier, $r(T_{k}^{\ast }(f_{i_{1}},0)\cdots (f_{i_{t}},0))\geq
k-t $. One can show that $r(T)\geq r(T_{k}^{\ast }(f_{i_{1}},0)\cdots
(f_{i_{t}},0))-(n-t)\geq k-t-n+t=k-n$.
\end{proof}

\begin{lemma}
\label{7M6} For any $k\in \omega \setminus \{0\}$, $h^{d}(T_{k})\geq k-1$.
\end{lemma}

\begin{proof}
If $k=1$, then the considered inequality holds. Let $k>1$. We now show that $%
h^{d}(T_{k}^{\ast })\geq k-1$. One can prove that there exists a
deterministic decision tree for the table $T_{k}^{\ast }$ such that
\begin{itemize}
\item $h(\Gamma )=h^{d}(T_{k}^{\ast })$.

\item For each node of $\Gamma $, which is neither the root nor a terminal
node, there are exactly two edges leaving this node.

\item For each complete path $\tau $ of $\Gamma $, $T_{k}^{\ast }(\tau )\neq
\Lambda $.
\end{itemize}

Let $\tau _{0}$ be a complete path in $\Gamma $ such that, for each node of $%
\tau _{0}$, which is neither the root nor a terminal node, the edge leaving
this node is labeled with the number $0$. Evidently, $r(T_{k}^{\ast }(\tau
_{0}))=1$. By Lemma \ref{7M5}, the number of nodes in $\tau _{0}$ labeled
with attributes is at least $k-1$. Therefore $h(\Gamma )\geq k-1$ and $%
h^{d}(T_{k}^{\ast })\geq k-1$. It is clear that each deterministic decision
tree for the table $T_{k}$ is a deterministic decision tree for the table $%
T_{k}^{\ast }$. Therefore $h^{d}(T_{k})\geq k-1$.
\end{proof}

The next statement follows immediately from Theorem 4.6 from \cite{Moshkov05}.

\begin{lemma}
\label{7M7} For any table $T\in \mathcal{M}%
_{2}^{\infty }$, if $T\neq \Lambda $, then $N(T)\leq (4W(T))^{Z(T)}$.
\end{lemma}

\begin{lemma}
\label{7M8} Let $A$ be a nontrivial closed class of decision tables from $\mathcal{M}%
_{2}^{\infty }$ and $\psi $ be a bounded complexity measure. If the function
$Z_{\psi ,A}$ is not everywhere defined, then the function $\mathcal{H}^{\infty}_{\psi ,A}$ is
not everywhere defined.
\end{lemma}

\begin{proof}
Let, for $n\in \omega $, the value $Z_{\psi ,A}(n)$ be not defined and $k\in
\omega \setminus \{0\}$. Then there exists a table $T\in $ $A_{\psi }(n)$ such that $%
Z(T)\geq m(k)$. It is easy to see that the set $[T]$ contains a complete
table $T^{\prime }$ with $m(k)$ columns. It is clear that $J(\nu
_{k},T^{\prime })=T_{k}$, i.e., $T_{k}\in A_{\psi }(n)$. Using Lemma \ref{7M4}%
, and boundedness from above property of the function $\psi $, we obtain
that $\psi ^{a}(T_{k})\leq 3n$. Using Lemma \ref{7M6}, and boundedness from
below property of the function $\psi $, we obtain that $\psi ^{d}(T_{k})\geq
k-1$. Since $k$ is an arbitrary number from $\omega \setminus \{0\}$, we obtain that the
value $\mathcal{H}^{\infty}_{\psi ,A}(3n)$ is not defined.
\end{proof}

\begin{lemma}
\label{7M9} Let $A$ be a nontrivial closed class of decision tables from $\mathcal{M}%
_{2}^{\infty }$ and $\psi $ be a bounded complexity measure. If the function
$\mathcal{H}^{\infty}_{\psi ,A}$ is everywhere defined, then, for any $n\in \omega $, $\mathcal{H}^{\infty}_{\psi
,A}(3n)\geq \sqrt{2Z_{\psi ,A}(n)}-3$.
\end{lemma}

\begin{proof}
If $Z_{\psi ,A}(n)=0$, then, evidently, $\mathcal{H}^{\infty}_{\psi ,A}(3n)\geq \sqrt{Z_{\psi
,A}(n)}-3$. Let $Z_{\psi ,A}(n)>0$ and $k$ be the maximum number from $%
\omega $ such that $\frac{k(k+1)}{2}\leq Z_{\psi ,A}(n)$. Then, $T_{k}\in
A_{\psi }(n)$ and, by Lemmas \ref{7M4} and \ref{7M6}, $\psi ^{a}(T_{k})\leq 3n$
and $\psi ^{d}(T_{k})\geq k-1$. Evidently, $\frac{(k+1)(k+2)}{2}>Z_{\psi
,A}(n)$. Therefore $(k+2)^{2}>2Z_{\psi ,A}(n)$ and $k>\sqrt{2Z_{\psi ,A}(n)}%
-2$. Thus, $\mathcal{H}^{\infty}_{\psi ,A}(3n)\geq \sqrt{2Z_{\psi ,A}(n)}-3$.
\end{proof}

\begin{lemma}
\label{7M10} Let $A$ be a nontrivial closed class of decision tables from $\mathcal{M}%
_{2}^{\infty }$, $\psi $ be a bounded complexity measure, $n\in \omega $, $%
T\in A_{\psi }(n)$ and $G(T)>0$. Then there exists a decision table $T^{\ast
}\in \lbrack T]$ such that $\psi ^{a}(T^{\ast })\leq n$ and $\psi
^{d}(T^{\ast })\geq G(T)-1$.
\end{lemma}

\begin{proof}
Let $n\in \omega $, $T\in A_{\psi }(n)$ and $G(T)>0$. Let $\alpha
=(f_{j_{1}},\sigma _{1})\cdots (f_{j_{m}},\sigma _{m})$ be an irreducible
annihilating word for the table $T$, which length is equal to $G(T)$ and in which the attributes $f_{j_{1}},\ldots ,f_{j_{m}}$ are arranged in the same order as in the table $T$. Denote
$Q=\operatorname{At}(T)\setminus \{f_{j_{1}},\ldots ,f_{j_{m}}\}$ and $T^{\prime }=$ $I(Q,T)$%
. Let us define a mapping $\nu :E_{2}^{m}\rightarrow \mathcal{P}(\omega )$.
Let $\bar{\delta}=(\delta _{1},\ldots ,\delta _{m})\in E_{2}^{m}$. Denote $%
\bar{\sigma}=(\sigma _{1},\ldots ,\sigma _{m})$. If $\bar{\delta}=\bar{\sigma%
}$, then $\nu (\bar{\delta})=\{0\}$. Let $\bar{\delta}\neq \bar{\sigma}$.
Then $\nu (\bar{\delta})=\{i: i\in\{1,\ldots ,m\},\delta _{i}\neq \sigma _{i}\}$. Denote $T^{\ast
}=J(\nu ,T^{\prime })$.

Let $\bar{\delta}=(\delta _{1},\ldots ,\delta _{m})\in \Delta(T^{\ast})$. Then there is $i\in\{1, \ldots , m\}$ such that $\delta _{i} \neq \sigma _{i}$. Evidently, $i$ is a common decision for the table $T^{\ast}(f_{j_{i}},\delta_i)$. Therefore $M(T^{\ast},\bar{\delta}) \le 1$. Using Lemma \ref{7M3}, we obtain $h^{a}(T^{\ast })\leq 1$.
Since the word $\alpha$ is irreducible, for each $i\in\{1, \ldots , m\}$, the set $\Delta(T^{\ast})$ contains a row (tuple) that is different from $\bar{\sigma}$ only in the $i$th digit and is labeled with the set of decisions $\{i\}$. Using this fact, it is easy to show that
$M(T^{\ast},\bar{\sigma}) \ge m-1$. Therefore $M(T^{\ast}) \ge m-1$.  According to Theorem 3.1 from \cite{Moshkov05}, $h^{d}(T^{\ast })\geq M(T^{\ast})$. Thus, $h^{d}(T^{\ast })\geq m-1=G(T)-1$. Using boundedness from above and
boundedness from below properties of the function $\psi $, we obtain $\psi
^{a}(T^{\ast })\leq n$ and $\psi ^{d}(T^{\ast })\geq G(T)-1$.
\end{proof}

\begin{lemma}
\label{7M11} Let $A$ be a nontrivial closed class of decision tables from $\mathcal{M}%
_{2}^{\infty }$ and $\psi $ be a bounded complexity measure. If the function
$G_{\psi ,A}$ is not everywhere defined, then the function $\mathcal{H}^{\infty}_{\psi ,A}$ is
not everywhere defined.
\end{lemma}

\begin{proof}
Let $n\in \omega $ and the value $G_{\psi ,A}(n)$ be not defined. Then there
exists an infinite sequence $D_{0},D_{1},\ldots $ of decision tables from $%
A_{\psi }(n)$ such that $0<G(D_{0})<G(D_{1})<\cdots \,$. Let $i\in \omega $.
From Lemma \ref{7M10} it follows that there exists a decision table $%
D_{i}^{\ast }\in \lbrack D_{i}]$ such that $\psi ^{a}(D_{i}^{\ast })\leq n$
and $\psi ^{d}(D_{i}^{\ast })\geq G(D_{i})-1$. Therefore the value $\mathcal{H}^{\infty}_{\psi
,A}(n)$ is not defined.
\end{proof}

\begin{lemma}
\label{7M12} Let $A$ be a nontrivial closed class of decision tables from $\mathcal{M}%
_{2}^{\infty }$ and $\psi $ be a bounded complexity measure. If the function
$\mathcal{H}^{\infty}_{\psi ,A}$ is everywhere defined, then, for any $n\in \omega $, $\mathcal{H}^{\infty}_{\psi
,A}(n)\geq G_{\psi ,A}(n)-1$.
\end{lemma}

\begin{proof}
Let $n\in \omega $. If $G_{\psi ,A}(n)=0$, then the considered inequality holds. Let $G_{\psi ,A}(n)>0$,  $T\in A_{\psi }(n)$ and $G(T)=G_{\psi ,A}(n)$. From
Lemma \ref{7M10} it follows that there exists a decision table $T^{\ast }\in
\lbrack T]$ such that $\psi ^{a}(T^{\ast })\leq n$ and $\psi ^{d}(T^{\ast
})\geq G(T)-1$. Therefore $\mathcal{H}^{\infty}_{\psi ,A}(n)\geq G_{\psi ,A}(n)-1$.
\end{proof}

\begin{lemma}
\label{7M13} Let $A$ be a nontrivial closed class of decision tables from $\mathcal{M}%
_{2}^{\infty }$, $\psi $ be a bounded complexity measure, the function $%
\mathcal{H}^{\infty}_{\psi ,A}$ be everywhere defined and $n\in \omega \setminus \{0\}$.
Then
$
\mathcal{H}^{\infty}_{\psi ,A}(n)\leq \max (n,nG_{\psi ,A}(n))Z_{\psi ,A}(n)\log _{2}(4nL_{\psi
,A}(n))$.

\end{lemma}

\begin{proof}
Let $T\in A$ and $\psi ^{a}(T)\leq n$. We now show that
\begin{equation*}
\psi ^{d}(T)\leq \max (n,nG_{\psi ,A}(n))Z_{\psi ,A}(n)\log _{2}(4nL_{\psi
,A}(n)).
\end{equation*}%
If $T\in \mathcal{M}_{2}^{\infty c}$, then $\psi ^{d}(T)=0$ and
the considered inequality holds. Let $T\notin \mathcal{M}_{2}^{\infty c}$. One can show that there exists a nondeterministic decision
tree $\Gamma $ for the table $T$ such that $\psi (\Gamma )\leq n$ and the
set $U$ of words corresponding to complete path of $\Gamma $ is an
irreducible $(\psi ,n)$-cover of the table $T$. Let $U=\{\alpha _{1},\ldots
,\alpha _{t}\}$ and $d_{1},\ldots ,d_{t}$ be decisions attached to terminal
nodes of complete paths with corresponding words $\alpha _{1},\ldots ,\alpha
_{t}$, respectively. Let, for the definiteness, $\operatorname{At}(\Gamma )=\{f_{1},\ldots ,f_{m}\}$ and attributes $f_{1},\ldots ,f_{m}$ are arranged in the same order as in the table $T$. It is
clear that $t\leq $ $L_{\psi ,A}(n)$. Using boundedness from below property
of the function $\psi $, we obtain $m\leq nL_{\psi ,A}(n)$. Denote $%
T^{\prime }=I(\operatorname{At}(T)\setminus \operatorname{At}(\Gamma ),T)$. One can show that $U$ is an
irreducible $(\psi ,n)$-cover of the table $T^{\prime }$. We now define a mapping $\nu
:E_{2}^{m}\rightarrow \mathcal{P}(\omega )$. Let $d_0 \in \omega \setminus \{d_{1},\ldots ,d_{t}\}$ and
$\bar{\delta}=(\delta
_{1},\ldots ,\delta _{m})\in E_{2}^{m}$. If $\bar{\delta} \notin \Delta (T^{\prime })$, then $\nu (\bar{\delta})=\{d_0\}$. Otherwise, $\nu (\bar{\delta})=\{d_i:\bar{\delta} \in \Delta (T^{\prime }\alpha_i), i \in \{1, \ldots ,t$\}\}.

 Denote $T^{\ast }=J(\nu ,T^{\prime })$.
One can show that any deterministic decision tree for the table $T^{\ast }$
is a deterministic decision tree for the table $T$. Therefore $\psi
^{d}(T)\leq \psi ^{d}(T^{\ast })$. It is clear that $U$ is an
irreducible $(\psi ,n)$-cover of the table $T^{\ast }$. Therefore $\psi
^{a}(T^{\ast })\leq n$. Using commutativity and nondecreasing properties of
the function $\psi $, we obtain that $\psi (f_{j})\leq n$ for any $f_{j}\in
\operatorname{At}(\Gamma )$. Therefore $T^{\ast }\in A_{\psi }(n)$. We shown that $%
W(T^{\ast})=m\leq nL_{\psi ,A}(n)$. Using Lemma \ref{7M7}, we obtain $N(T^{\ast
})\leq (4nL_{\psi ,A}(n))^{Z(T^{\ast})}\leq (4nL_{\psi ,A}(n))^{Z_{\psi ,A}(n)}$.
We now show that $$M_{\psi }(T^{\ast })\leq \max (n,nG_{\psi ,A}(n)).$$ Let $%
\bar{\delta}=(\delta _{1},\ldots ,\delta _{m})\in E_{2}^{m}$. First, we
consider the case when $\bar{\delta}\in \Delta (T^{\ast })$. Using Lemma \ref%
{7M3} and inequality $\psi^{a}(T^{\ast })\leq n$, we obtain $M_{\psi }(T^{\ast },\bar{\delta})\leq n$. Let $\bar{\delta}%
\notin \Delta (T^{\ast })$. Then the word $\alpha =(f_{1},\delta _{1})\cdots
(f_{m},\delta _{m})$ is an annihilating word for the table $T^{\ast }$. By removal
of some letters from the word $\alpha $, we obtain an irreducible
annihilating word $\beta $ for the table $T^{\ast }$. It is clear that the length of
$\beta $ is at most $G(T^{\ast })$ and $G(T^{\ast })\leq G_{\psi ,A}(n)$.
Let $\beta =(f_{j_{1}},\delta _{j_{1}})\cdots (f_{j_{s}},\delta _{j_{s}})$.
Using the fact that $T^{\ast }\in A_{\psi }(n)$ and boundedness from above
property of the function $\psi $, we obtain that $\psi (f_{j_{1}}\cdots
f_{j_{s}})\leq nG_{\psi ,A}(n)$. Therefore $M_{\psi }(T^{\ast },\bar{\delta})\leq
nG_{\psi ,A}(n)$. Hence $M_{\psi }(T^{\ast })\leq \max (n,nG_{\psi ,A}(n))$.
Using Lemma \ref{7L2}, we obtain $$\psi ^{d}(T)\leq \psi ^{d}(T^{\ast })\leq
\max (n,nG_{\psi ,A}(n))Z_{\psi ,A}(n)\log _{2}(4nL_{\psi ,A}(n)).$$ Thus, the statement of the lemma holds.
\end{proof}

\begin{proof}[Proof of Theorem \protect\ref{7T6}]
Using Theorem \ref{7T3} and Lemmas \ref{7M8} and \ref{7M11}, we obtain that the
functions $Z_{\psi ,A}$, $G_{\psi ,A}$ and $L_{\psi ,A}$ are everywhere
defined. Let there exist polynomials $p_{1}$, $p_{2}$, and $p_{3}$ such that
$Z_{\psi ,A}(n)\leq p_{1}(n)$, $G_{\psi ,A}(n)\leq p_{2}(n)$ and $L_{\psi
,A}(n)\leq 2^{p_{3}(n)}$ for any $n\in \omega $. Using Lemma \ref{7M13}, we
obtain that there exists a polynomial $p_{0}$ such that $\mathcal{H}^{\infty}_{\psi ,A}(n)\leq
p_{0}(n)$ for any $n\in \omega $. Let there be no a polynomial $p_{1}$ such
that $Z_{\psi ,A}(n)\leq p_{1}(n)$ for any $n\in \omega $ or there be no a
polynomial $p_{2}$ such that $G_{\psi ,A}(n)\leq p_{2}(n)$ for any $n\in
\omega $, or there be no a polynomial $p_{3}$ such that $L_{\psi ,A}(n)\leq
2^{p_{3}(n)}$ for any $n\in \omega $. Using Lemmas \ref{7M2}, \ref{7M9}, and %
\ref{7M12}, we obtain that there is no a polynomial $p_{0}$ such that $%
\mathcal{H}^{\infty}_{\psi ,A}(n)\leq p_{0}(n)$ for any $n\in \omega $.
\end{proof}

\section{Conclusions} \label{7S6}

In this paper, we looked at closed classes of decision tables with many-valued decisions. For decision tables from an arbitrary closed class, a function has been studied that characterizes the growth in the worst case of the minimum complexity of deterministic decision trees with an increase in the minimum complexity of nondeterministic decision trees. We found a criterion that this function is defined everywhere and studied the behavior of the function in the case when it satisfies this criterion. The results shed light on the relationship between deterministic decision trees and systems of decision rules represented by nondeterministic decision trees.

\subsection*{Acknowledgements}

Research reported in this publication was
supported by King Abdullah University of Science and Technology (KAUST).

\bibliographystyle{spmpsci}
\bibliography{abc_bibliography}

\end{document}